\documentclass[letterpaper, 10 pt, conference]{ieeeconf}  

\IEEEoverridecommandlockouts                              
\usepackage{mathtools}
\usepackage{amsmath,amssymb,amsfonts}
\usepackage{algorithm}
\usepackage{algpseudocode}
\usepackage{cite}

\newtheorem{theorem}{Theorem}
\newtheorem{lemma}{Lemma}
\newtheorem{assumption}{Assumption}
\newtheorem{definition}{Definition}
\newtheorem{remark}{Remark}
\newtheorem{corollary}{Corollary}
\usepackage{matlab-prettifier}
\usepackage{tikz}

\DeclareMathOperator*{\argmin}{arg\,min}

\let\labelindent\relax
\usepackage{enumitem}

\makeatletter
\newcommand{\specialcell}[1]{\ifmeasuring@#1\else\omit$\displaystyle#1$\ignorespaces\fi}
\makeatother

\newcommand\copyrighttext{%
  \footnotesize \copyright 2023 IEEE. This work has been accepted for publication at the 62nd IEEE Conference on Decision and Control (CDC), Singapore, 2023. Personal use of this material is permitted. Permission from IEEE must be obtained for all other uses, in any current or future media, including reprinting/republishing this material for advertising or promotional purposes, creating new collective works, for resale or redistribution to servers or lists, or reuse of any copyrighted component of this work in other works.}
\newcommand\copyrightnotice{%
\begin{tikzpicture}[remember picture,overlay]
\node[anchor=south,yshift=10pt] at (current page.south) {\fbox{\parbox{\dimexpr\textwidth-\fboxsep-\fboxrule\relax}{\copyrighttext}}};
\end{tikzpicture}%
}

\title{\LARGE \bf
On the Finite-Time Behavior of Suboptimal Linear\\
Model Predictive Control 

}

\author{Aren Karapetyan, Efe C. Balta, Andrea Iannelli and John Lygeros
\thanks{This work has been supported by the Swiss National Science Foundation under NCCR Automation (grant agreement  $51\text{NF}40\_180545$).}
\thanks{A. Karapetyan and J. Lygeros are with the Automatic Control Laboratory, ETH Z\"urich, 8092 Z\"urich, Switzerland
       (E-mails: {\tt\small \{akarapetyan, lygeros\}@control.ee.ethz.ch}).}%
\thanks{E. C. Balta  is  with  Inspire AG, 8005 Zürich, Switzerland (E-mail: {\tt\small efe.balta@inspire.ch}).}
\thanks{A. Iannelli is with the Institute for Systems Theory and Automatic Control, University of Stuttgart,  Stuttgart 70569,  Germany
        (E-mail: {\tt\small andrea.iannelli@ist.uni-stuttgart.de}).}
}

\begin{document}

\maketitle
\thispagestyle{empty}
\pagestyle{empty}

\copyrightnotice

\begin{abstract}

Inexact methods for model predictive control (MPC), such as real-time iterative schemes or time-distributed optimization, alleviate the computational burden of exact MPC by providing suboptimal solutions. While the asymptotic stability of such algorithms is well studied, their finite-time performance  has not received much attention. In this work, we quantify the performance of suboptimal linear model predictive control in terms of the additional closed-loop cost incurred due to performing only a finite number of optimization iterations. Leveraging this novel analysis framework, we propose a novel suboptimal MPC algorithm with a diminishing horizon length and finite-time closed-loop performance guarantees. This analysis allows the designer to plan a limited computational power budget distribution to achieve a desired performance level. We provide numerical examples to illustrate the algorithm's transient behavior and computational complexity.

\end{abstract}

\section{Introduction}
\label{sec:introduction}

Model Predictive Control (MPC) is one of the most ubiquitous optimal control methods thanks to its capability of handling state and input constraints and providing closed-loop performance guarantees \cite{kouvaritakis2016model}. The MPC algorithm relies on solving a constrained optimization problem at each sample time and hence requires a fast enough computational unit to handle it. In many practical examples, a short sampling time, combined with a large-scale optimization problem, can make the method infeasible to operate. This limitation has motivated the development of suboptimal MPC schemes for applications with limited computational capacity, e.g. \cite{zeilinger2011real,richter2011computational}. With suboptimal methods, at each time step, the optimization problem is solved approximately but continues to conform to certain performance requirements. While many works study the closed-loop stability of such methods, the characterization of transient performance in terms of the incurred cost has not been completely addressed. Such an analysis can be beneficial not only for certifying suboptimality bounds given a limited computational budget but also, as we show here, for designing time-varying MPC schemes, that  allocate computational budget adaptively.

The closed-loop stability of suboptimal MPC is well studied, e.g. in \cite{scokaert1999suboptimal, mcgovern1999closed, graichen2010stability, zeilinger2011real, rubagotti2014stabilizing}. Asymptotic stability is usually guaranteed by the consideration of a suitable Lyapunov function. In \cite{richter2011computational}, the authors derive a lower bound on the number of optimization iterations of a fast gradient method to achieve a certain suboptimality level of the MPC cost, however, the closed-loop stability is not analyzed.  
Time-distributed optimization \cite{liao2020time,liao2021analysis} or real-time iterative algorithms \cite{diehl2005nominal,zanelli2020lyapunov} are examples of another approach that considers the combined system-optimizer dynamics. These methods perform only a finite number of iterations of an optimization problem at each timestep. The asymptotic stability of  time-distributed MPC (\emph{TD-MPC}) is studied in \cite{liao2020time} for discrete-time non-linear models with state and input constraints, and an explicit form for a Lyapunov function is derived in \cite{zanelli2020lyapunov} for the same setting. Discrete-time linear models with a quadratic cost objective (LQMPC) are studied in \cite{liao2021analysis} and a region of attraction (ROA) estimate is derived in \cite{leung2021computable}.

In this work, we consider the transient performance of the suboptimal time-distributed optimization for LQMPC. In particular, our contribution is threefold. Firstly, we propose scheduling the number of iterative optimization steps, $\ell_k$, in advance, and allowing them to be time-varying based on the available computational budget.  Building on \cite{zanelli2020lyapunov,leung2021computable} we derive an explicit form for the rate of decay of the exponentially stable suboptimal dynamics under certain conditions on $\ell_k$-s. Secondly,  we study the  transient performance of this scheme by quantifying its {\em incurred suboptimality}, which we define as  the additional incurred closed-loop cost due to the approximate solution of the optimization problem. Finally, using our new analysis, we propose a diminishing horizon suboptimal MPC algorithm, called \emph{Dim-SuMPC}, and  quantify its finite-time performance. The proposed algorithm maintains the recursive feasibility and stability properties of \emph{TD-MPC} while decreasing the prediction horizon length at certain switching times. This decrease reduces the problem size and, hence, the time complexity. We provide numerical examples to illustrate the performance of the proposed scheme in terms of both cost and computational time.

\textit{Notation}: The set of positive real numbers is denoted by $\mathbb{R}_{+}$, and the set of non-negative integers by $\mathbb{N}$. For a given vector $x$, its  Euclidean norm is denoted by $\|x\|$, and the two-norm weighted by some matrix $Q\succ 0$  by $\|x\|_Q = \sqrt{x^{\top}Qx}$.  For a matrix $W \succ 0$ the spectral radius and the spectral norm are denoted by $\rho(W)$, and  $\|W\|$, respectively. Given $M\succ0$, the $\lambda_M^-(W)$ and $\lambda_M^+(W)$ denote the minimum and maximum eigenvalues of ${M}^{-\frac{1}{2}}W{M}^{-\frac{1}{2}}$, and recall that for any vector $x$, they satisfy $\lambda_M^{-}(W)\|x\|_M^2 \leq\|x\|_W^2 \leq \lambda_M^{+}(W)\|x\|_M^2$. The projection of a vector $x$ on a nonempty, closed convex set $\mathcal{A}$ is denoted by $\Pi_\mathcal{A}[x] := \argmin_{y\in \mathcal{A}}\|x-y\|$.

\section{Problem Formulation and Preliminaries}
\label{sec:problem_formulation}
We consider discrete-time linear time-invariant systems
\begin{equation*}
    x_{k+1} = Ax_k+Bu_k,
\end{equation*}
where $x_k\in \mathbb{R}^n$ and $u_k \in \mathbb{R}^m$ denote the state and control input at time $k$, respectively, $A\in \mathbb{R}^{n\times n}$, and $B\in\mathbb{R}^{n\times m}$. Given an initial state $x_0$, the control objective is to find the sequence of control inputs $\boldsymbol{u} = [u_0^{\top} \hdots u_{T-1}^{\top}]^{\top}$ that minimizes the finite-time cost
\begin{equation*}
    J_T(x_0,\boldsymbol{u}) = \|x_T\|^{2}_P+\sum_{k=0}^{T-1}\|x_k\|^{2}_Q+ \|u_k\|^2_R,
\end{equation*}
where $Q\in \mathbb{R}^{n \times n}$ and $R \in \mathbb{R}^{m\times m}$ are design matrices and $P$ is taken to be the solution of the discrete Algebraic Riccati Equation (DARE), $P = Q + K^\top R K + (A-BK)^\top P (A-BK)$, with $K = (R+B^\top PB)^{-1}(B^\top PA)$. In addition to the above optimality requirement, control inputs must satisfy
$u_k \in \mathcal{U}$ for all $k>0$ where  $\mathcal{U} \subseteq \mathbb{R}^m$ is a constraint set. We assume the following standard assumptions hold; these ensure that the problem is well posed, similar to \cite{mayne2000constrained}.

\begin{assumption}

\label{assum:well_posed}(Well-posed problem)
\begin{enumerate}[label=\roman*.] 
    \item The pair $(A,B)$ is stabilizable, $Q\succ 0$, $R\succ 0$.
    \item The input constraint set $\mathcal{U}$ is closed, convex, and contains the origin. 
    \label{assum:convex_u}
\end{enumerate}
\end{assumption}
At each timestep $0\leq k<T$, the model predictive controller solves the following parametric optimal control problem (POCP)
\begin{equation}
\label{eq:MPC_POCP}
\begin{split}
    \mu^{\star}(x_k) := &\argmin_{\boldsymbol{\nu}} \;  J_N(\xi_0,\boldsymbol{\nu})\\
     \text{s.t.} \; & \xi_{i+1} = A\xi_i +B\nu_i, \; i \!= \!0,\dots, N-1,\\
    &\xi_0 = x_k, \; \nu_i \in \mathcal{U}, \; i \!= \!0,\dots, N-1,
\end{split}
 \end{equation}
where $N$ is the prediction horizon length, and  $\boldsymbol{\nu} = [{\nu}_0^{\top} \hdots {\nu}_{N-1}^{\top}]^{\top}$ denotes the predicted input vector. The solution to \eqref{eq:MPC_POCP} for a given initial state (parameter) $x\in\mathbb{R}^n$, $\mu^{\star}(x):\mathbb{R}^n \rightarrow \mathbb{R}^{Nm}$, is referred to as the optimal mapping. The optimal cost attained by this mapping is denoted by $V_N(x):= J_N(x,\mu^{\star}(x))$, and serves as an approximate value function for the problem\footnote{When clear from the context, we drop the explicit dependence of $V_N(x)$ on $N$ and use $V(x)$.}. 

For each $k$, the first element of $\mu^{\star}(x_k)$ is applied to the system and the process is repeated in a receding horizon fashion. The optimal state evolution under this optimal MPC policy, starting from some $x_0^{\star}:=x_0$, is then given by
\begin{equation}
\label{eq:optimal_system}
    x^{\star}_{k+1} = Ax_k^{\star} +\overline{B}\mu^{\star}(x_k^{\star}):= f(x^{\star}_k), \; \forall k\geq 0,
\end{equation}
where , $S := \left[I_{m\times m}~\boldsymbol{0}~\hdots~\boldsymbol{0}\right] \in \mathbb{R}^{m \times Nm}$ and $\overline{B} := BS$.  

Problem \eqref{eq:MPC_POCP} is a parametric quadratic program, and, given a $x\in \mathbb{R}^n$, can be represented in an equivalent condensed form $V(x) =\min_{\boldsymbol{\nu}\in \mathcal{N}}  \|(x,\boldsymbol{\nu})\|_M^2$, with  $\mathcal{N} = \mathcal{U}^N \subseteq \mathbb{R}^{Nm}$
\begin{equation}
\label{eq:condensed_matrices}
    M = \begin{bmatrix}
W & G^{\top}\\
G & H
\end{bmatrix},
\end{equation}
with $H\in \mathbb{R}^{Nm\times Nm}$, $W\in \mathbb{R}^{Nm \times Nm}$ and,  $G \in \mathbb{R}^{Nm \times n}$ defined in Appendix \ref{sec:matrices_definition}.

As computing $\mu^{\star}(x)$ may be prohibitive, one can consider suboptimal solutions computed by  using a fixed number of optimization iterations. Specifically, given $x\in \mathbb{R}^n$ and an input vector $\boldsymbol{\boldsymbol{\nu}} \in \mathbb{R}^{Nm}$, consider the operator that performs one step of the projected gradient method (PGM)
\begin{equation}
\label{eq:pgm_operator}
    \mathcal{T}(x,\boldsymbol{\boldsymbol{\nu}}):= \Pi_{\mathcal{N}}[\boldsymbol{\nu} - \alpha \nabla{\boldsymbol{\nu}}{J_N}(x,\boldsymbol{\nu}) ],
\end{equation}
where $\alpha\in\mathbb{R}$ is a step size. Applying \eqref{eq:pgm_operator} iteratively for some  $\ell_k \in \mathbb{N}$ times, provides an approximation for the optimal input, and hence the optimal policy. The combined dynamics of the system and the optimizer are given by\footnote{The subscript of $\ell_k$ is dropped when it is taken to be a constant.}
\begin{subequations}
\label{eq:suboptimal_combined}
\begin{align}
  \label{eq:suboptimal_input}
    z_{k} &= \mathcal{T}^{\ell_k}(x_{k},z_{k-1}),\\
    x_{k+1} & = Ax_k + \overline{B}z_k,
    \label{eq:suboptimal_system}
\end{align}
\end{subequations}
where for some $l \in \mathbb{N}, x\in\mathbb{R}^n$ and $\boldsymbol{\nu}\in\mathbb{R}^{Nm}$, we define
\begin{equation*}
    \mathcal{T}^{l}(x,\boldsymbol{\nu}) = \mathcal{T}(x,\mathcal{T}^{l-1}(x,\boldsymbol{\nu})),
\end{equation*}
starting with $\mathcal{T}^0(x,\boldsymbol{\nu}) = \boldsymbol{\nu}$. The suboptimal state evolution \eqref{eq:suboptimal_system} is equivalently described by 
\begin{equation}
    x_{k+1}  = \underbrace{Ax_k + \overline{B}{\mu}^{\star} (x_k)}_{f(x_k)} + \overline{B} d(z_k,x_k),
    \label{eq:suboptimal_perturbed}
\end{equation}
where $d(z,x): = z - \mu^{\star}(x)$, $z \in \mathbb{R}^{Nm}, x\in \mathbb{R}^n$, is thought of as a disturbance acting on the optimal  dynamics \eqref{eq:optimal_system}, introduced due to suboptimality. We denote the combined system-optimizer state $s_k:= [x_k^\top~z_k^\top]^\top$ for all $k\geq0$.

In this work, we are interested in the transient performance of the closed-loop suboptimal dynamics \eqref{eq:suboptimal_perturbed}. This is captured in terms of the algorithm's incurred additional cost due to its approximations, i.e. the {\em incurred suboptimality}, defined as
\begin{equation}
    \mathcal{R}(x_0;\ell_{[0,\dots,T-1]}):= J_T(x_0,\boldsymbol{u}^{\ell}) - J_T(x_0,\boldsymbol{u}^{\star}),
    \label{eq:suboptimality_gap}
\end{equation}
where, $\boldsymbol{u}^{\ell}:= [z_0^\top\hdots z_{T-1}^\top]^\top$ denotes the sequence of inputs generated by the suboptimal policy on the suboptimal trajectory generated by \eqref{eq:suboptimal_perturbed}, and $\boldsymbol{u}^{\star}: = [\mu^{\star \top}(x_0)\hdots \mu_{T-1}^{\star \top}(x_{T-1}^\star)]^\top$ the equivalent for the optimal policy and the optimal trajectory generated by \eqref{eq:optimal_system}. Note that $\mathcal{R}$ is a function of $\ell_k$-s which are the main parameters indicating the level of approximation in the policy. Quantifying the transient accumulated suboptimality as a comparative metric due to suboptimal choices with respect to a more powerful benchmark is inspired by online learning, where a similar notion of regret is used, and the suboptimality is due to some uncertainty in the problem rather than limited computational budget. Although the system to be controlled is linear time-invariant, the closed loop systems of interest, \eqref{eq:optimal_system} and \eqref{eq:suboptimal_combined}, under the optimal and suboptimal policies, respectively, are nonlinear. Hence, our results make use of the standard notion of local exponential stability.
\begin{definition}
Consider a nonlinear autonomous system $x_{k+1} = f(x_k)$ for all $k\in \mathbb{N}$, with $f:\mathbb{R}^n \rightarrow\mathbb{R}^n$ having an equilibrium point $\hat{x} = \boldsymbol{0}$. The system is said to be locally exponentially stable  with a decay rate of $\lambda \in (0,1)$ if there exist $\gamma, c \in \mathbb{R}_{+}$, such that for all $k\geq 0$
\begin{equation*}
    \|x_k\| \leq \gamma \|x_0\|\lambda^k, \qquad \forall \: \|x_0\|\leq c.
\end{equation*} 
\end{definition}

\section{closed-loop Properties of MPC}

\subsection{Optimal MPC}
\label{sec:optimal_mpc}
In this subsection we review the properties of the optimal mapping $\mu^{\star}(x)$, which are then used for the analysis of the ``perturbed" suboptimal dynamics \eqref{eq:suboptimal_perturbed}.
We start with the regularity properties of the optimal mapping.
\begin{lemma} \cite[Corollary~2]{liao2021analysis}
\label{lem:lispchitz}
    For any $x,y\in\mathbb{R}^n$, the optimal solution mapping, $\mu^{\star}(x)$, satisfies
    \begin{align*}
        &\|\mu^{\star}(x)-\mu^{\star}(y)\|\leq \|H^{-\frac{1}{2}}\|\|G(x-y)\|_{H^{-1}} \leq L \|x-y\|,\\
        &\left<\mu^{\star}(x)-\mu^{\star}(y),G(x-y)\right>\leq -\|\mu^{\star}(x)-\mu^{\star}(y)\|_H^2,
    \end{align*}
    with a Lipschitz constant $L: = \|H^{-\frac{1}{2}}\|\cdot \|H^{-\frac{1}{2}}G\|$.
\end{lemma}

The proof follows from the parametric quadratic program structure of the MPC problem and can be found in \cite{liao2021analysis} or \cite{bemporad2002explicit} with an explicit MPC point of view.

As shown in \cite{limon2006stability, leung2021computable},  system \eqref{eq:optimal_system} is asymptotically stable with the ROA estimate
\begin{equation*}
    \Gamma_N: = \{x\in\mathbb{R}^n \mid \psi(x)\leq r_N\},
\end{equation*}
where $\psi(x):=\sqrt{V(x)}$, $\textstyle{d = c\cdot {\lambda^{-}(Q)}/{\lambda^{+}(P)}}$, $r_N = \sqrt{N d + c}$ and $c>0$ is such that the following set is non-empty
\begin{equation*}
    \Omega = \{x\in\mathbb{R}^n \mid \|x\|_P^2\leq c , -Kx \in \mathcal{U} \}.
\end{equation*}

\begin{lemma}
\label{lem:optimal_stability} 
    The value function $V$, is a local Lyapunov function for the optimal closed-loop system \eqref{eq:optimal_system} with a ROA estimate $\Gamma_N$, satisfying
    \begin{align}
        \label{eq:V_bounds}
        \|x\|_P^2 \leq V(x) \leq \|x\|^2_W\\
        \label{eq:exp_rate_optimal_v}
        V\left(f(x)\right) \leq \beta^2 V(x),
    \end{align}
with $\beta = \sqrt{1 - \lambda_W^{-}(Q)} \in (0,1)$ . System \eqref{eq:optimal_system} is then locally exponentially stable in $\Gamma_N$ with a rate of decay $\beta$.

\end{lemma}
\begin{proof}
The value function of the constrained MPC problem will necessarily attain a cost no less than the optimal infinite horizon unconstrained problem, i.e, for all $x\in \Gamma_N$ $\|x\|_P^2 \leq V(x)$.
Using the condensed form of \eqref{eq:MPC_POCP}
\begin{align*}
\label{eq:V_upper_bound}
    V(x) &= \|x\|_W^2 + 2\langle\mu^{\star}(x),Gx\rangle + \|\mu^{\star}(x)\|_H^2\\
    &\leq \|x\|_W^2 - \|\mu^{\star}(x)\|_H^2 \leq \|x\|_W^2,
\end{align*}
where the inequality follows from  Lemma \ref{lem:lispchitz}. For all $x \in \Gamma_N$ \begin{equation*}
     V\left(f(x)\right) - V(x)\leq -\|x\|_Q^2 \leq -\lambda_W^{-}(Q)\|x\|^2_W,
\end{equation*}
given the terminal cost matrix $P$ and quadratic stage costs \cite{mayne2000constrained}. Using the upper bound in \eqref{eq:V_bounds},
\begin{equation*}
    V(f(x)) \leq (1-\lambda_W^{-}(Q))V(x) = \beta^2V(x),
\end{equation*}
where $\lambda_W^{-}(Q)\in (0,1)$, as $W\succ Q$ from the structure of the matrices in Appendix \ref{sec:matrices_definition}. The function $V$ is then a Lyapunov function with a fixed rate of decay. Hence, by Lyapunov's direct method, $f(x)$ is locally exponentially stable \cite{haddad2011nonlinear}. In particular, for all $x_0 \in \Gamma_N$, given \eqref{eq:V_bounds}
\begin{align*}
    \|x_k^{\star}\|_P &\leq \psi(x_k^{\star})\leq \beta^k\psi(x_0) \Longrightarrow \\
    \|x_k^{\star}\| &\leq \frac{\beta^k\psi(x_0)}{\sqrt{\lambda^{-}(P)}}\leq\frac{\|x_0\|_W}{\sqrt{\lambda^{-}(P)}} \cdot \beta^k,
\end{align*}
showing the desired exponential decay rate.
\end{proof}

\subsection{Suboptimal MPC}
\label{sec:suboptimal_mpc}
To analyze the suboptimal closed-loop system \eqref{eq:suboptimal_perturbed}, we look at the combined system-optimizer dynamics \eqref{eq:suboptimal_combined}. The following theorem characterizes the linear convergence rate of PGM.

\begin{theorem}\cite[Theorem~3.1]{taylor2018exact} 
\label{the:pgm_contraction}
For any $x \in \mathbb{R}^n$, $\boldsymbol{\nu} \in \mathbb{R}^{N m}$, $\ell \in \mathbb{N}$, and for $\alpha = \frac{1}{\lambda^{+}(H)+\lambda^{-}(H)}$

$$
\left\|\mathcal{T}^{\ell}(x,\boldsymbol{\nu})-\mu^{\star}(x)\right\| \leq \eta^{\ell}\|\boldsymbol{\nu}-\mu^{\star}(x)\|,
$$
where $\eta=(\lambda^{+}(H) -\lambda^{-}(H)) /(\lambda^{+}(H)+\lambda^{-}(H))$.
\end{theorem}
\begin{remark}
We take the initial $\boldsymbol{\nu} = \boldsymbol{0}$.
\end{remark}
The stability of \eqref{eq:suboptimal_combined} is assessed by analysing the evolution of the suboptimality disturbance $d(z_k,x_k)$ and $\psi(x_k)$ over time.
\begin{lemma} 
\label{lem:auxiliary_dynamics}
    Given the dynamics \eqref{eq:suboptimal_combined}, for all $k\geq0$, $x_k \in \Gamma_N$ and $z_k\in \mathcal{N}$, the following holds
    \begin{align*}
         \psi(x_{k+1}) &\leq \beta \psi(x_k) + \sigma \|d(z_k,x_k)\|, \\
         \|d(z_{k+1},x_{k+1})\| &\leq \eta^{\ell_k}\kappa\psi(x_k) +\eta^{\ell_k}\omega\|d(z_k,x_k)\|,
    \end{align*}
    where  $\omega = 1 + \|H^{-\frac{1}{2}}\|\|H^{-\frac{1}{2}}G\overline{B}\|$,  $\sigma = \|W^{\frac{1}{2}}\overline{B}\|$, and
    \begin{align*}
    \kappa &=\|H^{-\frac{1}{2}}\|\|H^{-\frac{1}{2}}G(A-I)P^{-\frac{1}{2}}\|\\
    &+\|H^{-\frac{1}{2}}\|\sqrt{\lambda_H^+(G\overline{B})(\lambda_P^+(W)-1)}.
    \end{align*}
\end{lemma}
The lemma is a slightly modified version of \cite[Lemma~7]{leung2021computable} with time-varying $\ell_k$-s. The proof follows directly from the one found in \cite{leung2021computable}.
The following shows the existence of a Lyapunov function for the augmented state $s_k$.

\begin{theorem} 
\label{the:stability_theorem}
If $\ell_k>\ell^{\star}$ for all $k\geq 0$, where 
\begin{equation*}
    \ell^{\star} = \frac{\log(1-\beta) - \log(\sigma\kappa + \omega(1-\beta))}{\log(\eta)},
\end{equation*}
then  the system-optimizer dynamics \eqref{eq:suboptimal_combined} are asymptotically stable in the forward invariant ROA estimate
\begin{align*}
    \Sigma_N = \biggl\{ (x, z)\! \in\! \Gamma_N\! \times\! \mathcal{N} \mid~& \psi(x)\! \leq\! r_N, \biggr. \\ 
        &\biggl. \|z-\mu^{\star}(x)\| \leq \frac{(1-\beta)r_{N}}{\sigma} \biggl\}.
\end{align*} 
The function
\begin{equation}
    \mathcal{L}(x,z) := \psi(x) + \tau \|z-\mu^{\star}(x)\|
    \label{eq:lyapunov_function}
\end{equation}
is a local Lyapunov function for \eqref{eq:suboptimal_combined}, defined in $\Sigma_N$, given $\tau$ satisfies the following inequalities for all $k\geq 0$
\begin{equation}
\begin{split}
\label{eq:tau_equation}
    (\beta-1) + \tau\eta^{\ell_k}\kappa & <0,\\
    \sigma + (\eta^{\ell_k}\omega-1)\tau & <0,\\
    \tau > 0.
\end{split}
\end{equation}
\end{theorem}
For time-invariant $\ell$, the proof for the general case can be found in \cite{zanelli2020lyapunov} and for the LQMPC in \cite{leung2021computable}; the extension to time varying $\ell_k$-s follows directly.
\begin{remark}
We note that the terms $\beta, \omega, \kappa, \sigma$ and $\ell^{\star}$ depend on $N$ implicitly. We will make this dependence explicit, when unclear from the context.
\end{remark}

The following result provides a guaranteed rate of decay for the Lyapunov function \eqref{eq:lyapunov_function}, hence proving exponential stability of the combined dynamics \eqref{eq:suboptimal_combined}.

\begin{theorem}
  Given the dynamics \eqref{eq:suboptimal_combined}, for  all $s_0\in \Sigma_N$,  $k\geq0$ and  $\ell_k> \ell^{\star}$
\begin{subequations}
        \label{eq:exponentatial_stability_varying_l}
\begin{align*}
    \mathcal{L}(x_{k+1},z_{k+1}) &\leq \varepsilon_k \mathcal{L}(x_k,z_k),\\
    \mathcal{L}(x_{0},z_{0}) &\leq h_0 \cdot \|x_0\|_W,
    \end{align*}
\end{subequations}
with $\varepsilon_k := \max\{\beta+\tau\kappa\eta^{\ell_k}, \frac{\sigma+\tau\eta^{\ell_k}\omega}{\tau}\} \in (0,1)$, $\varepsilon_{-1}:=1$ and $h_0 = 1+\tau\eta^{\ell_0}L\|W^{-\frac{1}{2}}\|$.
\label{the:lyapunov_exponential_time_varying}
\end{theorem}
\begin{proof}
Recalling \eqref{eq:lyapunov_function}
\begin{align*}
    \mathcal{L}(x_{k+1},&z_{k+1}) = \psi(x_{k+1})\! +\! \tau\|d(z_{k+1}, x_{k+1})\| \\
    &\!\leq\! \underbrace{(\beta \!+\! \tau\eta^{\ell_k}\kappa)}_{\overline{\beta}_k}\psi(x_k)\! +\! \tau\underbrace{\left(\frac{\sigma +\tau\eta^{\ell_k}\omega}{\tau}\right)}_{\overline{\varepsilon}_k}\|d(z_k,x_k)\|\\
    &\leq \underbrace{\max\{\overline{\beta}_k,\overline{\varepsilon}_k\}}_{:=\varepsilon_k}\left(\psi(x_k)+\tau\|d(z_k,x_k)\|\right)\\
    &=\max\{\overline{\beta}_k,\overline{\varepsilon}_k\}\mathcal{L}(x_k,z_k),
\end{align*}
where the equalities follow from the definition of the Lyapunov function and the first inequality follows from Lemma \ref{lem:auxiliary_dynamics}. Note that ${\varepsilon}_k \in (0,1)$ as $\tau$ satisfies \eqref{eq:tau_equation} for all $k\geq 0$. 
For the second bound, consider
\begin{align*}
    \mathcal{L}(x_0,z_0) &= \psi(x_0) + \tau\|z_0 - \mu^{\star}(x_0)\|\\
    &\leq \|x_0\|_W + \tau\|\mathcal{T}^{\ell_0}(x_0,\boldsymbol{0}) - \mu^{\star}(x_0)\|\\
    &\leq \|x_0\|_W +\tau\eta^{\ell_0}\|\mu^\star(x_0)\|\\
     &\leq \|x_0\|_W\left(1+\tau\eta^{\ell_0}L\|W^{-\frac{1}{2}}\|\right),
\end{align*}
where the first inequality follows from Lemma \ref{lem:optimal_stability} and \eqref{eq:suboptimal_input}, the second from Theorem \ref{the:pgm_contraction} and the third from  Lemma \ref{lem:lispchitz}.
\end{proof}

The following corollary follows directly from the above theorem for a fixed $\ell>\ell^{\star}$.
\begin{corollary}
  Given the dynamics \eqref{eq:suboptimal_combined}, for  all $\ell> \ell^{\star}$, $s_0\in \Sigma_N$ and  $k\geq0$
\begin{equation*}
    \mathcal{L}(x_{k+1},z_{k+1})\leq \varepsilon \mathcal{L}(x_k,z_k),
\end{equation*}
where $\mathcal{L}$ is defined in \eqref{eq:lyapunov_function} and $\varepsilon := \max\{\beta+\tau\kappa\eta^{\ell}, \frac{\sigma+\tau\eta^{\ell}\omega}{\tau}\} \in (0,1)$.
\label{the:lyapunov_exponential}
\end{corollary}

While the exponential stability of suboptimal MPC has been pointed out in \cite{zanelli2020lyapunov}, its rate of decay has not been explicitly derived. The above results define this rate, which becomes the main tool used for the finite-time analysis of the algorithm in the next section.

\section{Finite-time Analysis}

The exponential stability of the combined dynamics \eqref{eq:suboptimal_combined} allows one to study the finite-time performance of the suboptimal LQMPC. In this section, we quantify the incurred suboptimality \eqref{eq:suboptimality_gap} of suboptimal dynamics \eqref{eq:suboptimal_combined}  both for varying and for fixed number of optimization iterations $\ell$. Before proceeding with the suboptimality analysis we introduce some auxiliary lemmas.

\begin{lemma}
  Given the dynamics \eqref{eq:suboptimal_combined}, for  all $s_0\in \Sigma_N$,  $k\geq0$ and  $\ell_k> \ell^{\star}$
\begin{equation*}
    \|\Delta \mu_k\| : = \|z_k - \mu^{\star}(x_k^{\star})\|\leq b_0\|x_0\|_W\prod_{i=-1}^{k-1}\varepsilon_{i}+c\|x_0\|_W\beta^k,
\end{equation*}
where $c = \max\{\tau^{-1},\|H^{-\frac{1}{2}}\|\|H^{-\frac{1}{2}}GP^{-\frac{1}{2}}\|\}$, $b_0\!=\!c \cdot h_0$ and $h_0$ is defined as in Theorem \ref{the:lyapunov_exponential_time_varying}.
\label{lem:delta_input_varying_l}
\end{lemma}

\begin{proof}
The proof follows by showing that $\|\Delta \mu_k\|$ is upper bounded by the the sum of two Lyapunov functions and enjoys the same rate of decay as those. In particular, for any $k\geq0$ and $s_k\in\Sigma_N$
\begin{align*}
    \|\Delta \mu_k\| &= \|z_k  - \mu^{\star}(x_k) + \mu^{\star}(x_k) - \mu^{\star}(x_k^{\star})\|\\
    & \leq \|z_k  - \mu^{\star}(x_k)\| + \|\mu^{\star}(x_k) - \mu^{\star}(x^{\star}_k)\|\\
    & \leq \tau^{-1} \tau \|d(z_k,x_k) \| + {\|H^{-\frac{1}{2}}\|} \|H^{-\frac{1}{2}}G(x_k-x^{\star}_k)\|\\
    &\leq \tau^{-1} \tau \|d(z_k,x_k) \| \\
    & \qquad+ {\|H^{-\frac{1}{2}}\|}\|H^{-\frac{1}{2}}G P^{-\frac{1}{2}}\|(\|x_k\|_P + \|x_k^{\star}\|_{P}) \\
    & \leq c \left(\underbrace{\tau \|d(z_k,x_k)\| + \|x_k\|_P}_{:=l(x_k,z_k)} + \underbrace{\|x_k^{\star}\|_P}_{:=g(x_k^{\star})}\right),
\end{align*}
where the first and third inequalities make use of the triangle inequality, and the second one follows from Lemma \ref{lem:lispchitz}. From Lemma \ref{lem:optimal_stability} and \eqref{eq:lyapunov_function}, we note that $l(x_k,z_k) \leq \mathcal{L}(x_k,z_k)$ and  $g(x_k^{\star}) \leq  \psi(x_k^{\star})$, and it holds that
\begin{align*}
    \|\Delta \mu_k\| \leq c \left(\mathcal{L}(x_k,z_k) + \psi(x_k^{\star})\right).
\end{align*}
The result follows by a repeated application of \eqref{eq:exp_rate_optimal_v} and \eqref{eq:exponentatial_stability_varying_l}. 
\end{proof}

For a fixed  $\ell>\ell^{\star}$ the following corollary follows.

\begin{corollary}
Given the dynamics \eqref{eq:suboptimal_combined}, for  all  $\ell> \ell^{\star}$, $s_0\in \Sigma_N$ and  $k\geq0$
\begin{equation*}
\|\Delta \mu_k\| \leq b\|x_0\|_W\varepsilon^k+c\|x_0\|_W\beta^k,
\end{equation*}
where $h = 1+\tau\eta^{\ell}L\|W^{-\frac{1}{2}}\|$ and $b= c \cdot h$.
\label{lem:delta_input}
\end{corollary}

Next, we show that the exponential stability of the combined system-optimizer dynamics \eqref{eq:suboptimal_combined} implies  the same rate for the state, $x_k$, in the ROA estimate  $\Gamma_N$.

\begin{lemma}
Given the dynamics \eqref{eq:suboptimal_combined},  for  all $s_0\in \Sigma_N$,  $k\geq0$ and  $\ell_k> \ell^{\star}$
\begin{equation*}
    \|x_k\| \leq   h_0\|P^{-\frac{1}{2}}\|\cdot\|x_0\|_W\prod_{i=-1}^k\varepsilon_i,
\end{equation*}
where $h_0$ is defined as in Theorem \ref{the:lyapunov_exponential_time_varying}.
\label{lem:x_upper_bound_varying_l}
\end{lemma}
\begin{proof} From Theorem \ref{the:stability_theorem}, $\Sigma_N$ is a forward invariant ROA estimate for the suboptimal dynamics. Hence  from Lemma \ref{lem:optimal_stability} we have that for all $s_0\in \Sigma_N$ 
\begin{equation*}
    \|x_k\|_P \leq \psi(x_k)  + \tau\|z_k - \mu^\star(x_k)\| = \mathcal{L}(x_k,z_k),
\end{equation*}
and
\begin{align*}
    \|x_k\| \leq \|P^{-\frac{1}{2}}\|\mathcal{L}(x_0,z_0)\!\prod_{i=-1}^k\varepsilon_i\! \leq   h_0\|P^{-\frac{1}{2}}\|\|x_0\|_W\prod_{i=-1}^k\varepsilon_i,
\end{align*}
where we use the submultiplicative property of norms and Theorem \ref{the:lyapunov_exponential_time_varying}.
\end{proof}

\begin{corollary}
 Given the dynamics \eqref{eq:suboptimal_combined}, for  all  $\ell> \ell^{\star}$, $s_0\in \Sigma_N$ and  $k\geq0$
\begin{equation*}
    \|x_k\| \leq h\|P^{-\frac{1}{2}}\|\cdot\|x_0\|_W \cdot  \varepsilon^k,
\end{equation*}
where $h$ is defined as in Corollary \ref{lem:delta_input}.
\label{lem:x_upper_bound}
\end{corollary}
\subsection{Incurred Suboptimality}
The incurred suboptimality of the suboptimal dynamics \eqref{eq:suboptimal_combined} can  be bounded using the bounds in Lemmas \ref{lem:delta_input_varying_l} and \ref{lem:x_upper_bound_varying_l}.

\begin{theorem}
\label{the:suboptimality_varying_ell}
Given the dynamics \eqref{eq:suboptimal_combined}, for  all $s_0\in \Sigma_N$,  $k\geq0$ and  $\ell_k> \ell^{\star}$, its incurred suboptimality  is bounded by
\begin{equation*}
\mathcal{R}(x_0;\ell_{[0,\hdots,T-1]}) \leq \overline{c} \|x_0\|^2_W\cdot \sum_{k=0}^{T}\prod_{i=0}^{k}\varepsilon^2_{i-1} \leq \frac{\overline{c} \|x_0\|^2_W}{1-\overline{\varepsilon}^2},
\end{equation*}
where
\begin{equation*}
\begin{split}
\overline{c} = \max\{ &\|\overline{R}\|\left(b_0+c\right) \left((b_0+c)+\frac{2L}{\sqrt{\lambda^{-}(P)}}\right),  \\
&  \|\overline{Q}\| \cdot \left(\|P^{-\frac{1}{2}}\|^2h_0^2 + \frac{1}{{\lambda^{-}(P)}}\right)\},
\end{split}    
\end{equation*}
$\overline{R} = S^\top R S$, $\|\overline{Q}\| = \max\{\|Q\|,\|P\|\}$ and $\overline{\varepsilon} := \max\{\beta+\tau\eta^{\overline{\ell}}\kappa, \frac{\sigma+\tau\eta^{\overline{\ell}}\omega}{\tau}\}$ with $\overline{\ell} = \min_{k}\{\ell_k\}_{k=0}^{T-1}$.
\end{theorem}
\begin{proof}
    Using the quadratic form of the cost, expanding the squares, collecting the terms, and using the submultiplicative property of the norms, it can be shown that the incurred suboptimality \eqref{eq:suboptimality_gap} attains the following upper bound
\begin{equation*}
\begin{split}
     \mathcal{R}(x_0,\ell_{[0,\dots,T-1]})&\leq\underbrace{\sum_{k=0}^T \|x_k\|_{\overline{Q}}^2  + \|x_k^{\star}\|^2_{\overline{Q}}}_{:=s_1}\\
  &+\underbrace{\sum_{k=0}^{T-1} \|\Delta \mu_k\|^2_{\overline{R}} + 2\|\Delta \mu_k\|\|\mu^{\star}(x_k^{\star})\|\!\cdot\!\|\overline{R}\|}_{:=s_2},
\end{split}
\end{equation*}
where $\Delta \mu_k$ is defined as in Lemma \ref{lem:delta_input_varying_l}. First, we bound the term due to the input difference
\begin{align*}
    s_2&\leq \|\overline{R}\|\sum_{k=0}^{T-1}\|\left(\Delta \mu_k\|^2+2L\|\Delta \mu_k\|\|x_k^\star\|\right)\\
    &\leq\|\overline{R}\|\sum_{k=0}^{T-1}\left(\|\Delta \mu_k\|^2 +2\frac{L\beta^k\|x_0\|_W}{\sqrt{\lambda^{-}(P)}}\|\Delta {\mu_k}\|\right) \\
    &\leq \|\overline{R}\| \|x_0\|^2_W \left((b_0+c)^2+\frac{2L\left(b_0+c\right)}{\sqrt{\lambda^{-}(P)}}\right)\cdot\sum_{k=0}^{T}\prod_{i=-1}^{k-1}\varepsilon^2_i,
\end{align*}
where the first inequality follows from Lemma \ref{lem:lispchitz}, the second from Lemma \ref{lem:optimal_stability} and the third from Lemma \ref{lem:delta_input_varying_l}, and the fact that $\beta < \varepsilon_k$ for all $k\geq0$. Using the latter fact and the bounds in Lemmas \ref{lem:optimal_stability} and \ref{lem:x_upper_bound_varying_l}, it follows that
\begin{align*}
    s_1\leq \|\overline{Q}\| \cdot \|x_0\|^2_W \left(\|P^{-\frac{1}{2}}\|^2h_0^2 + \frac{1}{{\lambda^{-}(P)}}\right) \sum_{k=0}^{T}\prod_{i=-1}^{k-1}\varepsilon^2_i.
\end{align*}
The intermediate bound follows by combining the above two bounds for $s_1$ and $s_2$. The finite upper bound follows from geometric series with the appropriate definition of $\overline{\varepsilon}$.
\end{proof}

For a constant number of optimization iterations per timestep, the following corollary follows.
\begin{corollary}
\label{the:suboptimality_fixed_ell}
Given the dynamics \eqref{eq:suboptimal_perturbed}, for  all  $\ell> \ell^{\star}$, $s_0\in \Sigma_N$ and  $k\geq0$, its  incurred suboptimality is bounded by
\begin{equation*}
\mathcal{R}(x_0;\ell) <
\frac{\overline{c}\|x_0\|_W^2}{1-\varepsilon^2},
\end{equation*}
where $\overline{c}$ is given in Theorem \ref{the:suboptimality_varying_ell} and $\varepsilon := \max\{\beta+\tau\kappa\eta^{\ell}, \frac{\sigma+\tau\eta^{\ell}\omega}{\tau}\} \in (0,1)$.
\end{corollary}

To the best of our knowledge this is the first analysis that explicitly characterises the \emph{incurred suboptimality} of suboptimal LQMPC in terms of the closed-loop cost. While the bounds are conservative due to the lack of further assumptions on the system, these results can motivate the design of time-varying suboptimal MPC schemes for applications with limited computational budget. In the next section,  we present an example of such an algorithm.

\section{Dim-SuMPC}
The finite-time analysis in the previous section motivates the development of a novel MPC algorithm with a diminishing horizon length. The aim of the proposed method is to maintain the asymptotic and finite-time properties of existing suboptimal MPC methods while reducing the time required to solve the problem. In particular, the novel algorithm reduces the prediction horizon length $N$ of the suboptimal LQMPC a pre-defined $p\in\mathbb{N}$ number of times. This is a design parameter which can be set depending on the available computational budget. For each $p$, we define a sequence $\{N_j\}_{j=0}^p$, such that, $N_0:=N$ defined in Section \ref{sec:problem_formulation} and $N_{j-1}>N_j$ for all $1\leq j\leq p$. The following lemma provides the number of timesteps required to transition from a forward invariant ROA estimate defined for $N_{j-1}$ to a smaller one, defined for $N_j$.

\begin{lemma}
\label{lem:k_j}
    Given the dynamics \eqref{eq:suboptimal_combined},  if for any $1\leq j \leq p$, it holds that $N_{j-1}>N_j$ and $\ell>\ell^{\star}(N_j)$, then for all $s_0\in\Sigma_{N_{j-1}}$ and $k\geq k_j$, where
    \begin{equation}
    \label{eq:k_js}
        k_j\!=\!\frac{\log\left(\lambda_{W_{j-1}}^-(P)\left(N_jd+c\right)\right)-2\log\left(h(N_j)\|x_0\|_{{W_{j-1}}}\right)}{2\log(\varepsilon)},
    \end{equation}
   and $W_j:= W(N_j)$, the following holds
    \begin{equation}
       V_{N_j}(x_k)  \leq {N_jd+c}.
       \label{eq:dim_sumpc_value_function}
    \end{equation}
\end{lemma}
\begin{proof}
    As $\Sigma_{N_{j-1}}$ is a forward invariant region for the suboptimal dynamics \eqref{eq:suboptimal_combined}, it follows from Lemma \ref{lem:optimal_stability} that
    \begin{equation*}
        V_{N_{j-1}}(x_k)\leq \|x_k\|^2_{W_{j-1}}.
    \end{equation*}
Then, from Corollary \ref{lem:x_upper_bound}
    \begin{equation*}
       \sqrt{\lambda_{W_{j-1}}^{-}(P)} \|x_k\|_{W_{j-1}}\leq\|x_k\|_P \leq \varepsilon^kh(N_{j-1})\|x_0\|_{W_{j-1}}.
    \end{equation*}
By noting from the principle of optimality that $V_{N_j}(x_k) \leq V_{N_{j-1}(x_k)}$, the following is then sufficient to have the inequality in \eqref{eq:dim_sumpc_value_function} hold for all $1\leq j\leq p$
    \begin{equation*}
       \varepsilon^k \cdot\frac{h(N_{j-1})\cdot\|x_0\|_{W_{j-1}}}{\sqrt{\lambda^{-}_{W_{j-1}}(P)}}\leq \sqrt{N_jd+c}.
    \end{equation*}
    This is equivalent to the condition
    \begin{equation*}
        k\geq \frac{\log\!\left(\lambda_{W_{j-1}}^-(P)\left(N_jd+c\right)\right)\!-\!2\log\!\left(h(N_{j-1})\|x_0\|_{{W_{j-1}}}\right)}{2\log(\varepsilon)}.
    \end{equation*}
\end{proof} 
From the above lemma, it follows that, for example,  after $k=k_1$ steps of suboptimal dynamics  evolution with $\ell>\ell^{\star}(N_0)$ updates, the state $x_{k_1}$ will be in a new ROA estimate, $\Gamma_{N_1} = \{x\in \mathbb{R}^n:V_{N_1}(x)\leq N_1d +c\}$. At this point, the MPC problem \eqref{eq:MPC_POCP} can be redefined with the new horizon length $N_1<N_0$ and it can be  solved to optimality from that point on if the computational power allows so. This can then be repeated for all $j$. Note that the computational effort to solve for $\mu^{\star}_{N_j}$ is strictly less than that for the original optimal problem since $N>N_j$ for all $j$. In this new region, the redefined MPC with the reduced horizon length can be solved also suboptimally. In particular, it follows from Theorem \ref{the:stability_theorem} and Corollary \ref{the:lyapunov_exponential} that for each $j = 1,\dots,p$ and $N_j$ if $\ell>\ell^{\star}(N_j)$, the suboptimal dynamics are exponentially stable in the corresponding ROA estimate $\Sigma_{N_j}$. This motivates our proposed diminishing horizon suboptimal MPC scheme, Dim-SuMPC, outlined in Algorithm
\ref{alg:dim_sumpc}.

\emph{DimSuMPC} maintains the recursive feasibility property of \emph{TD-MPC} since the updates on the prediction horizon, from some $N_{j-1}$ to $N_{j}$ are such that the state always remains within a corresponding ROA estimate $\Gamma_{N_{j}}$. Thus, while the regulation/tracking performance of the two schemes is expected to be comparable, the computational time of \emph{DimSuMPC} is expected to be lower, as demonstrated on a numerical example in the next section.

\begin{algorithm}
\caption{Dim-SuMPC}\label{alg:dim_sumpc}
\begin{algorithmic}[1]
\State Fix the sequence $\{N_j\}_{j=0}^p$ and set $j=0$
\State Calculate $\{k_j\}_{j=1}^p$, according to \eqref{eq:k_js}
\State Take any $s_0 = [x_0^\top~z_0^\top]^\top \in \Sigma_{N_0}$
\For{$k = 0, \dots,T-1$}
    \If{$k \geq k_{j+1}$}
         \State  $j\gets j+1$
    \EndIf
\State Set $\ell_k > \ell^{\star}(N_j)$ and compute $z_k$ according to \eqref{eq:suboptimal_input}
\State Apply $u_k = Sz_k$ and get $x_{k+1}$ according to \eqref{eq:suboptimal_system}
\EndFor
\end{algorithmic}
\end{algorithm}
Guidelines on how to choose a valid initial point in step $2$ of the algorithm are outlined in \cite{leung2021computable}. The parameter $p$  and the horizon length sequence are design parameters and can be chosen in advance based on the capacity of the available budget. The incurred suboptimality of the algorithm is bounded in the following theorem.

\begin{theorem}
   The incurred suboptimality of the \emph{Dim-SuMPC} algorithm is bounded by
    \begin{equation*}
\mathcal{R}(x_0,\ell_{[0,\hdots,T]}) \leq \frac{\overline{c}_m\|x_0\|_{W_0}^2}{1-\underline{\varepsilon}}\cdot\sum_{j=0}^{p}\varepsilon_{k_{j}}^{2k_{j}}\prod_{i=1}^{j}\overline{d}_i,
\end{equation*}
where $k_0=0$, $\overline{d}_i: =h^2(N_i) \lambda^+(W_i)\|P^{-\frac{1}{2}}\|^2$, $\overline{c}_m: = \max_j \overline{c}(N_j)$, $\overline{h}:= \max_j h(N_j)$ and $\underline{\varepsilon}:= \max_j \varepsilon_{k_j}$.
\end{theorem}
\begin{proof}
    It follows from  the result in Theorem \ref{the:suboptimality_varying_ell}
    \begin{equation*}
        \mathcal{R}(x_0,\ell_{[0,\hdots,T]}) \leq \frac{\overline{c}_m}{1-\underline{\varepsilon}} \sum_{j =0}^p\|x_{k_j}\|^2_{W_j}.
    \end{equation*}
    Then, using the bound in Corollary \ref{lem:x_upper_bound}, for all $j=1,\dots,p$
    \begin{align*}
        \|x_{k_{j}}\|^2_{W_j}&\leq\\
        &h^2(N_j) \lambda^+(W_j)\|P^{-\frac{1}{2}}\|^2 \|x_{k_{j-1}}\|^2_{W_{j-1}}\varepsilon_{k_j}^{2(k_j-k_{j-1})}.
    \end{align*}
    The result follows by repeated application of the above.
    \end{proof}
    
Note that the above bound  is finite since $p$ is finite and the state remains in a forward invariant ROA set at all times.

\section{Numerical Examples}

In this section, we consider the following linearised, continuous-time model  of an inverted pendulum from \cite{leung2021computable}
\begin{equation*}
    A_c=\left[\begin{array}{cc}
0 & 1  \\
\frac{3g}{2L} & 0 \\
\end{array}\right],\: B_c =\left[\begin{array}{c}
0 \\
\frac{3}{{m_b}L^2} 
\end{array}\right],
\end{equation*}
where the state is $x= [\theta,~\dot{\theta}]^\top$, $\theta$ is the angle relative to the unstable equilibrium position and the control input is the applied torque. The parameters are taken to be the same as in \cite{leung2021computable} with $L=1$, $m_b=0.1$ and $g=9.81$. We consider the control of the discretized model of the plant with a sampling time of $T_s=0.1$. The input constraint set is taken to be $\mathcal{U}= [-1,1]$, the cost matrices are  $Q =  I_2$, and $R = 1$ and the initial state is $x_0 = [-\pi/4~\pi/5]^{\top}$. We demonstrate the performance of \emph{DimSuMPC} for this setting, over a control horizon of length $T=150$. In the first example, we  perform $p=3$ switches at times $k_1=15,\; k_2=25$ and $k_3=40$, sequentially decreasing the prediction horizon length from the initial $N=15$ to, respectively, $N_1=10, N_2=8$ and $N_3=2$. 
\begin{figure}
\begin{center}
\includegraphics[width=\columnwidth]{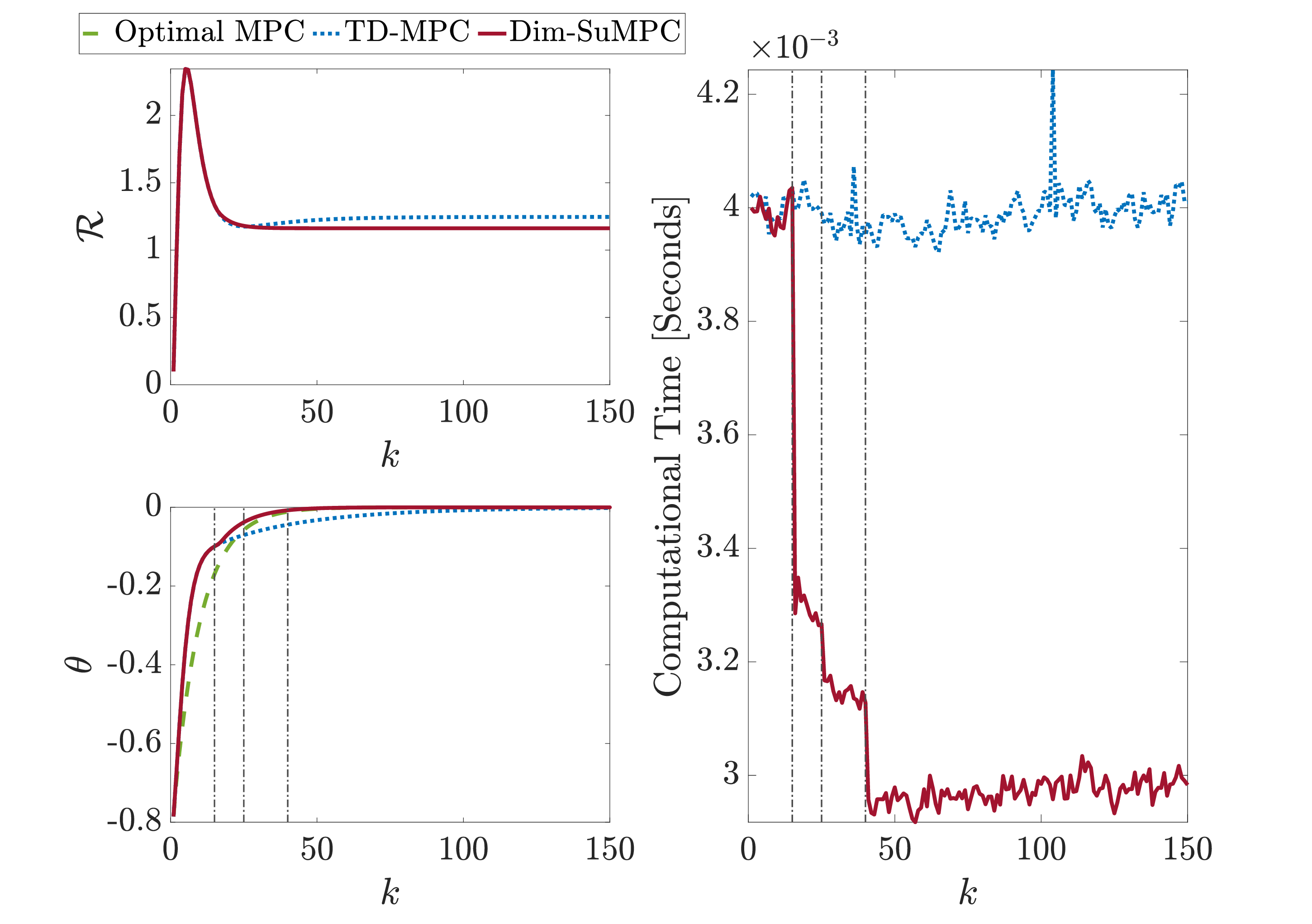}    
\caption{ \emph{Dim-SuMPC} is compared to \emph{TD-MPC} and optimal MPC. The plot on the right shows the reduction in computation time for \emph{Dim-SuMPC}, the top left one contains the incurred suboptimality of both suboptimal methods and the bottom left one shows the convergence of $\theta$. The switching times are marked by vertical gray lines.}
\label{fig:main}
\end{center}
\end{figure}
We let the number of optimization iterations to be fixed at $\ell = 5000$ (for both \emph{TD-MPC} and \emph{DimSuMPC}) to assess the effect of the diminishing horizon length. Note that $\ell^{\star}(N)$ and the switching times $k_j$-s from Lemma \ref{lem:k_j} are over-conservative in practice, and  we use smaller values in the examples. Figure \ref{fig:main} compares the closed-loop performance of  the optimal MPC with that of \emph{TD-MPC} and \emph{DimSuMPC}. As can be seen in the bottom left plot, the proposed scheme shows a comparable convergence performance to \emph{TD-MPC} and even suffers a lower  incurred suboptimality. As the problem size becomes smaller, the decrease in computational time of \emph{DimSuMPC} can be clearly observed at the switching times, marked by gray vertical lines on the right side figure. This saved time allows one to perform more iterative updates $\ell$. To demonstrate this, consider a second example where only a finite computational budget is available that allows the execution of \emph{TD-MPC} with $N=15$ and $\ell=5000$. Figure \ref{fig:fixed_budget} demonstrates how \emph{DimSuMPC} can achieve better convergence using (roughly) the same computational power by changing at $k_1$ to new parameters $N_1=2$ and $\ell=6500$. In other words, decreasing $N$ allows for an increase in $\ell$ resulting in a potentially improved performance while staying within the same computational budget. In both examples, the computation time is measured by taking the average of $200$ runs of the same experiment using the  \lstinline[style=Matlab-editor]{tic/toc} command in MATLAB. 
\begin{figure}
\begin{center}
\includegraphics[width=\columnwidth]{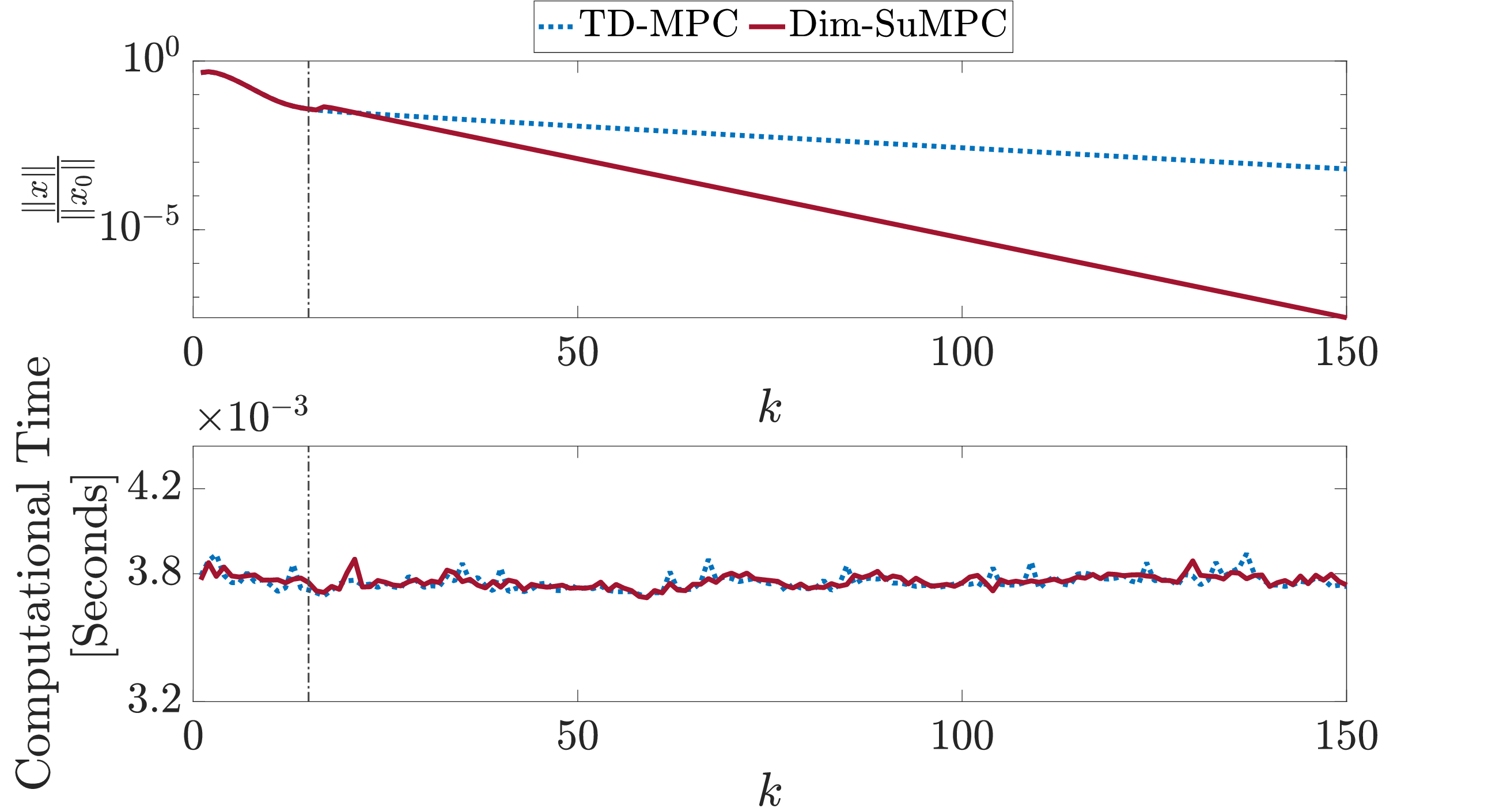}    
\caption{ \emph{DimSuMPC} spends the same computational effort as \emph{TD-MPC} (bottom plot) while achieving a faster convergence of the state (top plot).}
\label{fig:fixed_budget}
\end{center}
\end{figure}
\section{Conclusions}

We establish an explicit expression for the rate of convergence of a closed-loop system under a suboptimal implementation of the LQMPC algorithm subject to input constraints, where only a finite number of iterative optimization steps are performed using the projected gradient descent method. The bound is used to provide finite-time performance guarantees of the scheme in terms of the additional cost incurred due to suboptimality. A novel diminishing horizon suboptimal MPC algorithm is then proposed that operates by decreasing the prediction horizon length at certain switching times and thus reduces computational complexity. Possible directions for future research include a deeper analysis of the proposed suboptimal scheme and the derivation of less conservative bounds for the incurred suboptimality by exploring the properties of optimal MPC.


\begin{appendix}
\subsection{System Matrices for the POCP Problem}
\label{sec:matrices_definition}

As also shown in \cite{liao2021analysis}, the matrices in \eqref{eq:condensed_matrices} are given by $H=\hat{B}^\top \hat{H} \hat{B}+\left(I_N \otimes R\right)$, $G=\hat{B}^\top \hat{H} \hat{A}, W=Q+\hat{A}^\top \hat{H} \hat{A},\\ \hat{H}=\left[\begin{array}{cc}\left(I_N \otimes Q\right) & 0 \\ 0 & P\end{array}\right]$, 
\begin{equation*}
\hat{B}=\left[\begin{array}{ccc}
0 & 0 & 0 \\
B & 0 & 0 \\
\vdots & \ddots & \vdots \\
A^{N-1} B & \cdots & B
\end{array}\right], \text { and } \hat{A}=\left[\begin{array}{c}
I \\
A \\
\vdots \\
A^N
\end{array}\right] .
\end{equation*}

\end{appendix}
\section*{Acknowledgements}
The authors thank Dominic Liao-McPherson for fruitful insights and discussions on the topic.

\bibliographystyle{ieeetr}
\bibliography{bibliography.bib}

\end{document}